\documentclass[journal,twoside]{IEEEtranTCOM}

\normalsize

\IEEEoverridecommandlockouts

\usepackage{cite}

\usepackage{stfloats}
\usepackage{graphicx}
\usepackage{subfigure}
\usepackage{boxedminipage}
\usepackage{comment}

\usepackage{adjustbox}
\usepackage{array}
\newcolumntype{R}[2]{%
    >{\adjustbox{angle=#1,lap=\width-(#2)}\bgroup}%
    l%
    <{\egroup}%
}
\newcommand*\rot{\multicolumn{1}{R{45}{1em}}}

\usepackage{epstopdf}
\usepackage[cmex10]{amsmath}
\usepackage{amsmath}
\usepackage{amsthm}

\usepackage{algorithmic}

\usepackage{url}

\usepackage[all]{xy}

\renewcommand{\P}{\mathrm{P}}
\renewcommand{\b}{b}	
\newcommand{\p}{p}	
\newcommand{\f}{\mathrm{f}}

\newcommand{\I}{\mathrm{I}}

\newtheorem{theorem}{Theorem}

\usepackage{times,color}

\begin{document}
%
\title{Reliability-output Decoding of \\ Tail-biting Convolutional Codes}

%

\author{Adam~R.~Williamson,~\IEEEmembership{Student Member,~IEEE,}
        Matthew~J.~Marshall,
        and~Richard~D.~Wesel,~\IEEEmembership{Senior~Member,~IEEE}
\thanks{A. R. Williamson, M. J. Marshall, and R. D. Wesel are with the Electrical Engineering Department, University of California, Los Angeles, CA 90095 USA (e-mail: adamroyce@ucla.edu; mmarshall@ucla.edu; wesel@ee.ucla.edu).}
\thanks{This material is based upon work supported by the National Science Foundation under Grant Number 1162501. Any opinions, findings, and conclusions or recommendations expressed in this material are those of the author(s) and do not necessarily reflect the views of the National Science Foundation.}}

%
%

\maketitle


%
%

\begin{abstract}
We present extensions to Raghavan and Baum's reliability-output Viterbi algorithm (ROVA) to accommodate tail-biting convolutional codes. These tail-biting reliability-output algorithms compute the exact word-error probability of the decoded codeword after first calculating the posterior probability of the decoded tail-biting codeword's starting state. One approach employs a state-estimation algorithm that selects the maximum a posteriori state based on the posterior distribution of the starting states.
Another approach is an approximation to the exact tail-biting ROVA that estimates the word-error probability. A comparison of the computational complexity of each approach is discussed in detail.
The presented reliability-output algorithms apply to both feedforward and feedback tail-biting convolutional encoders.
These tail-biting reliability-output algorithms are suitable for use in reliability-based retransmission schemes with short blocklengths, in which terminated convolutional codes would introduce rate loss. 
\end{abstract}


\IEEEpeerreviewmaketitle

%
%
\section{Introduction}
\label{sec:intro}

Raghavan and Baum's reliability-output Viterbi algorithm (ROVA) \cite{Raghavan_ROVA_TransIT_1998} uses the sequence-estimation property of the Viterbi algorithm to calculate the exact word-error probability of a received convolutional code sequence. In general, the ROVA can be used to compute the word-error probability for any finite-state Markov process observed via memoryless channels (i.e., processes with a trellis structure). However, the ROVA is only valid for processes that terminate in a known state (usually  the all-zeros state). For codes with large constraint lengths $(\nu + 1)$, a significant rate penalty is incurred due to the $\nu$ additional symbols that must be transmitted in order to arrive at the termination state.

Tail-biting convolutional codes can start in any state, but must terminate in the same state.  The starting/terminating state is unknown at the receiver.  These codes do not suffer the rate loss of terminated codes, making them throughput-efficient (see, e.g., \cite{Ma_Wolf_TB_codes_TCOM_1986} and \cite[Ch. 12.7]{Lin_2004_ECC}). The tail-biting technique is commonly used for short-blocklength coding.

%
\subsection{Overview and Contributions}

In this paper, we extend the ROVA to compute the word-error probability for tail-biting codes. First, we present a straightforward approach, which we call the tail-biting ROVA ({\tt TB ROVA}). ~{\tt TB ROVA} invokes the original ROVA for each of the possible starting states $s$.  The complexity of this straightforward approach is large, proportional to  $2^{2\nu}$ for standard binary convolutional codes (and $q^{2\nu}$ for convolutional codes over Galois field $GF(q)$).   

We explore several approaches to reduce the complexity of {\tt TB ROVA}. We first introduce a post-decoding algorithm that computes the reliability of codewords that have already been decoded by an existing tail-biting decoder, including possibly suboptimal decoders. We then propose a new tail-biting decoder that uses the posterior distribution of the starting states to identify the most probable starting state of the received sequence.
Finally, we discuss how to use Fricke and Hoeher's simplified (approximate) ROVA \cite{Fricke_Approx_ROVA_VTC_2007} for each of the $q^\nu$ initial states, which reduces the complexity of the word-error probability computation.

The reliability-output algorithms presented in this paper apply to both feedforward (non-recursive) and feedback (recursive) convolutional encoders. However, as pointed out by St\r{a}hl et al. \cite{Stahl_Feedback_TB_TransIT_2002}, it is not possible to have a one-to-one mapping from information words to codewords and still fulfill the tail-biting restriction for feedback encoders at certain tail-biting codeword lengths. St\r{a}hl et al. \cite{Stahl_Feedback_TB_TransIT_2002} provide conditions for when tail-biting will work for recursive encoders and also describe how to determine the starting state corresponding to an input sequence. In the cases that the tail-biting technique works for feedback encoders, there is a one-to-one mapping between input sequences and codewords, and the reliability-output algorithms in this paper are valid.

The remainder of this paper proceeds as follows: Sec. \ref{sec:relatedliterature} reviews the related literature and Sec. \ref{sec:notation} introduces notation.  Sec.~\ref{sec:ROVA} reviews Raghavan and Baum's ROVA and discusses how to extend it to tail-biting codes. 
The simplified ROVA for tail-biting codes is discussed in Sec.~\ref{sec:simplified}. 
Sec.~\ref{sec:post-proc_rova} presents the Post-decoding Reliability Computation ({\tt PRC}) for tail-biting codes and Sec.~\ref{sec:state_estimation} introduces the Tail-Biting State-Estimation Algorithm ({\tt TB SEA}).
Sec.~\ref{sec:tbbcjr} discusses an alternative to {\tt TB SEA} using the tail-biting BCJR algorithm.
Sec.~\ref{sec:complexity} evaluates the complexity of the proposed algorithms, and
Sec.~\ref{sec:simulation_results} shows numerical examples of the computed word-error probability and the actual word-error probability.
Sec. \ref{sec:conc} concludes the paper.

%
\subsection{Related Literature}
\label{sec:relatedliterature}

There are a number of reliability-based decoders for terminated convolutional codes, most notably the Yamamoto-Itoh algorithm \cite{Yamamoto_Itoh_Algo_TransIT_1980}, which computes a reliability measure for the decoded word, but not the exact word-error probability. In \cite{Fricke_Reliability_HARQ_TCOM_2009}, Fricke and Hoeher use Raghavan and Baum's ROVA in a reliability-based type-I hybrid Automatic Repeat reQuest (ARQ) scheme.

Hof et al. \cite{Hof_Conv_Bounds_ISCTA_2009} modify the Viterbi algorithm to permit generalized decoding according to Forney's generalized decoding rule \cite{Forney_Erasure_TransIT_1968}.
When the generalized decoding threshold is chosen for maximum likelihood (ML) decoding with erasures and the erasure threshold is chosen appropriately, this augmented Viterbi decoder is equivalent to the ROVA.

A type-II hybrid ARQ scheme for incremental redundancy using punctured terminated convolutional codes is presented by Williamson et al. \cite{Williamson_ROVA_ISIT_2013}. In \cite{Williamson_ROVA_ISIT_2013}, additional coded symbols are requested when the word-error probability computed by the ROVA exceeds a target word-error probability. This word-error requirement facilitates comparisons with recent work in the information theory community \cite{Polyanskiy_CCR_2010,Polyanskiy_IT_2011_NonAsym}. Polyanskiy et al. \cite{Polyanskiy_IT_2011_NonAsym} investigate the maximum rate achievable at short blocklengths with 
variable-length feedback codes. While \cite{Williamson_ROVA_ISIT_2013} shows that terminated convolutional codes can deliver throughput above the random-coding lower bound of \cite{Polyanskiy_IT_2011_NonAsym}, the rate loss from termination is still significant at short blocklengths. To avoid the termination overhead, it is imperative to have a reliability-output decoding algorithm for tail-biting codes.

In contrast to the decoding algorithms for terminated codes, Anderson and Hladik \cite{Anderson_TB_MAP_JSAC_1998} present a tail-biting maximum a posteriori (MAP) decoding algorithm. This extension of the BCJR algorithm \cite{BCJR_TransIT_1974} can be applied to tail-biting codes with a priori unequal source data probabilities. As with the BCJR algorithm,  \cite{Anderson_TB_MAP_JSAC_1998} computes the posterior probabilities of individual data symbols.  In contrast, the ROVA \cite{Raghavan_ROVA_TransIT_1998} and the tail-biting reliability-based decoders in this paper compute the posterior probabilities of the codeword. 

More importantly, the tail-biting BCJR of \cite{Anderson_TB_MAP_JSAC_1998} is only an approximate symbol-by-symbol MAP decoder, as pointed out in \cite{Anderson_TB_BCJR_book_2001} and \cite{Anderson_BD_CVA_TCOM_2002}. Because the tail-biting restriction is not strictly enforced, non-tail-biting ``pseudocodewords" can cause bit errors, especially when the ratio of the tail-biting length $L$ to the memory length $\nu$ is small (i.e., $L/\nu \approx 1$-$2$). Further comparisons with the tail-biting BCJR are given in Sec.~\ref{sec:tbbcjr}. An exact symbol-by-symbol MAP decoder for tail-biting codes is given in \cite[Ch. 7]{Johannesson_Fundamentals_Conv_1999}.

Handlery et al. \cite{Handlery_Boosting_TCOM_2003} introduce a suboptimal, two-phase decoding scheme for tail-biting codes that computes the approximate posterior probabilities of each starting state and then uses the standard BCJR algorithm to compute the posterior probabilities of the source symbols. This approach is compared to the tail-biting BCJR of \cite{Anderson_TB_MAP_JSAC_1998} and exact MAP decoding in terms of bit-error-rate (BER) performance.  Both the two-phase approach of \cite{Handlery_Boosting_TCOM_2003} and the tail-biting BCJR of \cite{Anderson_TB_MAP_JSAC_1998} perform close to exact MAP decoding when $L/\nu$ is large, but suffer a BER performance loss when $L/\nu$ is small.

Although it does not compute the word-error probability, Yu \cite{Yu_State_Est_VTC_2008} introduces a method of estimating the initial state of tail-biting codes, which consists of computing a pre-metric for each state based on the last $\nu$ observations of the received word. This pre-metric is then used to initialize the path metrics of the main tail-biting decoder (e.g., the circular Viterbi decoder \cite{Cox_Circular_Viterbi_TransVT_1994}), instead of assuming that all states are equally likely at initialization. The state-estimation method of \cite{Yu_State_Est_VTC_2008}, which is not maximum-likelihood, is limited to systematic codes and a special configuration of nonsystematic codes that allows information symbols to be recovered from noisy observations of coded symbols.

Because tail-biting codes can be viewed as circular processes \cite{Ma_Wolf_TB_codes_TCOM_1986,Cox_Circular_Viterbi_TransVT_1994}, decoding can start at any symbol.  Wu et al. \cite{Wu_RB_TB_VTC_2010} describe a reliability-based decoding method that  compares the log likelihood-ratios of the received symbols in order to determine the most reliable starting-location for tail-biting decoders.  
Selecting a reliability-based starting location reduces the error probability by minimizing the chance of choosing non-tail-biting paths early in the decoding process. Wu et al. \cite{Wu_RB_TB_VTC_2010} apply this approach to existing suboptimal decoders, including the wrap-around Viterbi algorithm of \cite{Shao_TwoTB_TCOM_2003}.
As with \cite{Yu_State_Est_VTC_2008}, \cite{Wu_RB_TB_VTC_2010} does not compute the word-error probability.

Pai et al. \cite{Pai_HARQ_TBCC_TransVT_2011} generalizes the Yamamoto-Itoh algorithm to handle tail-biting codes and uses the computed reliability measure as the retransmission criteria for hybrid ARQ. When there is a strict constraint on the word-error probability, however, this type of reliability measure is not sufficient to guarantee a particular undetected-error probability.
Providing such a guarantee motivates the word-error probability calculations in this paper (instead of bit-error probability as in  \cite{Anderson_TB_MAP_JSAC_1998,Yu_State_Est_VTC_2008,Handlery_Boosting_TCOM_2003,Johannesson_Fundamentals_Conv_1999,Anderson_TB_BCJR_book_2001,Anderson_BD_CVA_TCOM_2002}).

%
\subsection{Notation}
\label{sec:notation}

We use the following notation in this paper: $\P(X = x)$ denotes the probability mass function (p.m.f.) of discrete-valued random variable $X$ at value $x$, which we also write as $\P(x)$. The probability density function (p.d.f.) of a continuous-valued random variable $Y$ at value $y$ is $\f(Y=y)$, sometimes written as $\f(y)$. In general, capital letters denote random variables and lowercase letters denote their realizations.
Boldface letters with superscripts denote vectors, as in ${\bf y}^\ell = (y_1, y_2, \dots, y_\ell)$, while subscripts denote a particular element of a vector: $y_i$ is the $i$th element of ${\bf y}^\ell$. 
We use the hat symbol to denote the output of a decoder, e.g.,  $\hat {\bf x}$ is the codeword chosen by the Viterbi algorithm.

%
%
\section{The Reliability-output Viterbi Algorithm}
\label{sec:ROVA}

Raghavan and Baum's reliability-output Viterbi algorithm \cite{Raghavan_ROVA_TransIT_1998} augments the canonical Viterbi decoder with the computation of the word-error probability of the maximum-likelihood (ML) codeword. In this section, we provide an overview of the ROVA.

For rate-$k/n$ convolutional codes with $L$ trellis segments and input alphabet $q$, we denote the ML codeword as $\hat {\bf x}^L = \hat {\bf x}$ and the noisy received sequence as ${\bf y}^L = {\bf y}$. The probability that the ML decision is correct given the received word is $\P(\bf X = \hat x | Y = y) = \P(\hat x | y )$, and the word-error probability is $\P({\bf X \neq \hat x | Y = y}) = 1 - \P({\bf \hat x | y })$. The probability of successfully decoding can be expressed as follows:
\begin{align}
	\P(\bf \hat x |y) &= \frac{\f(\bf y|\hat x) \P(\hat x)}{\f(\bf y)} 
	= \frac{\f(\bf y|\hat x) \P(\hat x)}{\sum \limits_{\bf x'} \f(\bf y|x') \P(x')},
\end{align}
where we have used $\f(\bf y|\hat x)$ to denote the conditional p.d.f. of the real-output channel (e.g., the binary-input AWGN channel). This may be replaced by the conditional p.m.f. $\P(\bf y|\hat x) $ for discrete-output channels (e.g., the binary symmetric channel).

The probability of correctly decoding can be further simplified if each of the codewords $\bf x'$ is a priori equally likely, i.e., $\P(\bf \hat x) = \P(x') ~\forall~ x' \neq \hat x$, which we assume for the remainder of the paper. This assumption yields
\begin{align}
	\P( \bf \hat x|y) &= \frac{\f(\bf y|\hat x) }{\sum \limits_{\bf x'} \f(\bf y|x') }.
	\label{eqn:P_x_y}
\end{align}
In general, the denominator in \eqref{eqn:P_x_y} may be computationally intractable when the message set cardinality is large. 
However, the ROVA \cite{Raghavan_ROVA_TransIT_1998} takes advantage of the trellis structure of convolutional codes to compute $\P(\bf \hat x|y)$ exactly with complexity that is linear in the blocklength and exponential in the constraint length of the code (i.e., it has complexity on the same order as that of the original Viterbi algorithm). This probability can also be computed approximately by the simplified (approximate) ROVA \cite{Fricke_Approx_ROVA_VTC_2007}, which will be discussed further in Sec.~\ref{sec:simplified}.

The ROVA can compute the probability of word error for any finite-state Markov process observed via memoryless channels (e.g., in maximum-likelihood sequence estimation for signal processing applications). In the remainder of this paper, we use the example of convolutional encoding and decoding, but the ROVA and our tail-biting trellis algorithms apply to any finite-state Markov process.


%
\subsection{Conditioning on the Initial State}

Raghavan and Baum's ROVA applies only to codes that begin and end at a known state.  Each of the probabilities $\f(\bf y|x')$ in \eqref{eqn:P_x_y} is implicitly conditioned on the event that the receiver knows the initial and final state of the convolutional encoder. 

To be precise, ROVA beginning and ending at the same state $s$, which we shall denote as {\tt ROVA}$(s)$, effectively computes the following:
\begin{align}
	\underbrace{ \P({\bf  \hat x}_s |{\bf y},s)}_\text{computed by {\tt ROVA}$(s)$} &= \frac{\f({\bf y}|{\bf \hat x}_s,s) \P({\bf \hat x}_s|s) }{\f({\bf y}|s) } &= \frac{\f({\bf y}|{\bf \hat x}_s,s) }{\sum \limits_{{\bf x}'_s} \f({\bf y}|{\bf x}_s',s) },
	\label{eqn:P_x_y_s}
\end{align}
where the limit ${\bf x}'_s$ in the summation of the denominator denotes that the enumeration for the summation is over all codewords ${\bf x}'_s$ with starting state $s$, and $\f({\bf y|x}_s',s)$ is shorthand for $\f({\bf Y=y|X=x}_s',S=s)$. 
In summary, {\tt ROVA}$(s)$ computes the ML codeword $\hat {\bf x}_s$ corresponding to starting state $s$, the posterior probability of that codeword given $s$, $\P(\hat {\bf x}_s | {\bf y}, s)$, and the probability of the received sequence given $s$, $\f({\bf y} | s)$. The inputs and outputs of {\tt ROVA}$(s)$ are illustrated in the block diagram of Fig.~\ref{fig:rova_block}.
%

\begin{figure}[htb]
\centering \def\svgwidth{290pt}
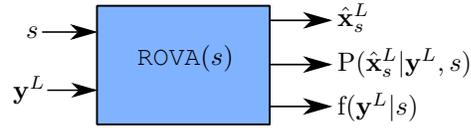
\caption{Block diagram of Raghavan and Baum's {\tt ROVA}$(s)$ \cite{Raghavan_ROVA_TransIT_1998}.}
\label{fig:rova_block}
\end{figure}

For tail-biting codes, we are interested in computing the quantity $\P(\bf \hat x|y)$ without conditioning on the unknown starting and ending state $s$:
\begin{align}
	\P(\bf \hat x|y) &= \sum \limits_{s} \P({\bf \hat x|y},s) \P(s|{\bf y}).
\end{align}
%
%
%
The ML codeword $\bf \hat x$ has an associated initial state, $\hat s$. Note that $\P({\bf \hat x|y},s) = 0$ unless $s=\hat s$, since $\bf \hat x$ is not a possible codeword for any starting state other than $\hat s$. Thus, we have:
\begin{align}
	\P(\bf \hat x|y) &= \P({\bf \hat x|y},\hat s) \P(\hat s| \bf y).
	\label{eqn:P_x_y_weighted}
\end{align}
Thus, the tail-biting ROVA ({\tt TB ROVA}) must compute the probability $\P(\bf \hat x|y)$ of successful decoding in \eqref{eqn:P_x_y_weighted} by weighting $ \P({\bf \hat x|y},\hat s) $ with $\P(\hat s| \bf y)$. (For the original ROVA with a known starting state $\hat s$, $\P(\hat s| {\bf y}) = 1$ and $\P(s'| {\bf y}) = 0 ~\forall ~s' \neq \hat s$.)

Using the fact that each of the initial states $s$ is equally likely {\it a priori} (i.e., $\P(s) = \P(s') ~\forall s \neq s'$), we have:
\begin{align}
	\P(\hat s| \bf y) &= \frac{\f({\bf y}|\hat s)}{ \sum \limits_{s'}  \f({\bf y}|s') }.
	\label{eqn:P_s_y}
\end{align}
This finally yields
\begin{align}
	\underbrace{\P(\bf \hat x|y)}_\text{computed by {\tt TB ROVA}} &= \frac{ \overbrace{\P({\bf \hat x|y},\hat s) \f({\bf y}|\hat s)}^\text{computed by {\tt ROVA$(\hat s)$}} }{ \sum \limits_{s'}  \underbrace{\f({\bf y}|s')}_\text{computed by {\tt ROVA$(s')$}} },
	\label{eqn:P_x_y_TB}
\end{align}
where the summation in the denominator of \eqref{eqn:P_x_y_TB} is over all $q^\nu$ possible initial states.

%
\subsection{A Straighforward Tail-biting ROVA}
\label{sec:straightforward_algo}

A straightforward ML approach to decoding tail-biting codes is to perform the Viterbi algorithm {\tt VA}$(s)$, for each possible starting state $s = 0, 1, \dots, q^{\nu}-1$. The ML codeword $\bf \hat x$ is then chosen by determining the starting state with the greatest path metric (i.e., the greatest probability). As shown in Fig.~\ref{fig:tb_rova_block}, this approach will work for the ROVA as well: perform {\tt ROVA}$(s)$ for each possible $s$ and then pick $\bf \hat x$ and its starting state $\hat s$. The probability $\P(\bf \hat x|y)$ is then computed as in \eqref{eqn:P_x_y_TB}, using $\P({\bf \hat x|y},\hat s)$ from the ROVA for the ML starting state and the $\f({\bf y}|s)$ terms produced by the ROVAs for all the states. This approach is illustrated in the block diagram of Fig.~\ref{fig:tb_rova_block}.
\begin{figure}[t]
\small
\centering 
\def\svgwidth{240pt}
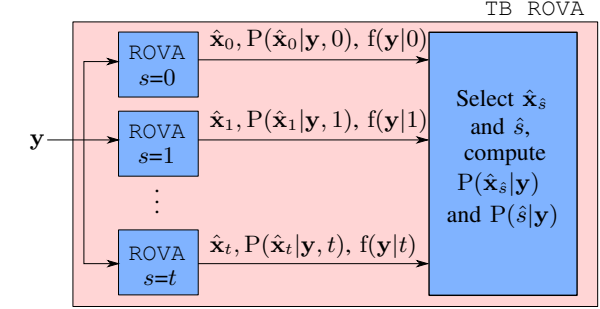
\caption{Block diagram of the straightforward tail-biting ROVA ({\tt TB ROVA}), which performs the {\tt ROVA($s$)} for each possible starting state $s$. The largest possible state is $t = q^\nu - 1$.}
\label{fig:tb_rova_block}
\end{figure}

%
%
\section{The Simplified ROVA for Tail-biting Codes}
\label{sec:simplified}

This section proposes replacing the exact word-error computations of Sec.~\ref{sec:straightforward_algo}'s straightforward {\tt TB ROVA} with an estimated word-error probability,  using Fricke and Hoeher's simplified (approximate) ROVA \cite{Fricke_Approx_ROVA_VTC_2007}. This approach requires a Viterbi decoder for each starting state to select the ML codeword for that state.
Fricke and Hoeher's \cite{Fricke_Approx_ROVA_VTC_2007} simplified ROVA for starting state $s$, which we call {\tt Approx ROVA$(s)$}, estimates the probability $\P({\bf x}_s|{\bf y},s)$. Substituting this estimate $\tilde \P({\bf x}_s|{\bf y},s)$ into \eqref{eqn:P_x_y_TB}, we have the following approximation:
\begin{align}
	\tilde \P(\bf \hat x|y) &= \frac{\overbrace{\tilde \P({\bf \hat x|y}, \hat s)}^\text{computed by {\tt Approx ROVA}$(\hat s)$} \f({\bf y}| \hat s)}{ \sum \limits_{s}  \f({\bf y}|s) } \label{eqn:ptilde}.
\end{align}
While $\f({\bf y}|s)$ is not computed directly by {\tt Approx ROVA}$(s)$, we can approximate it with quantities available from {\tt Approx ROVA}$(s)$ as
\begin{align}
	\f({\bf y}|s) \approx \tilde \f({\bf y}|s)&=\frac{ \f({\bf y|x}_s,s) \P({\bf x}_s| s) }{ \tilde \P({\bf x}_s | {\bf y}, s)} \\
	&= \frac{ \f({\bf y|x}_s, s) }{ q^{kK} \mathrm{\tilde P}({\bf x}_s | {\bf y}, s)} \label{eqn:ftilde},
\end{align}
when all $q^{kK}$ codewords with starting state $s$ are equally likely, where $K = L - \nu$ is the number of trellis segments before termination. Note $\f({\bf y|x}_s,s)$ can be calculated by the Viterbi algorithm for starting state $s$. 

Equations \eqref{eqn:ptilde} and \eqref{eqn:ftilde} lead to the following estimate of the word-correct probability:
\begin{align}
	{\tilde \P(\bf \hat x|y)} &\approx \frac{ \overbrace{\tilde \P({\bf \hat x|y},\hat s) \tilde \f({\bf y}|\hat s)}^\text{computed by {\tt Approx ROVA}$(\hat s)$} } { \sum \limits_s \underbrace{ \tilde \f({\bf y}|s)}_\text{computed by {\tt Approx ROVA}$(s)$ }}.	\label{eqn:ptilde_approx}
\end{align}
We refer to the overall computation of  ${\tilde \P(\bf \hat x|y)}$ in \eqref{eqn:ptilde_approx} as {\tt Approx TB ROVA}. Sec.~\ref{sec:complexity} provides a discussion of its complexity and Sec.~\ref{sec:simulation_results} presents simulation results.

Note that despite the approximations, the simplified ROVA chooses the ML codeword for terminated codes. For the tail-biting version {\tt Approx TB ROVA}, as long as the winning path metric of each starting/ending state is used to determine the ML state $\hat s$, the decoded codeword will also be ML (and the same as the codeword chosen by the exact tail-biting ROVA in Sec.~\ref{sec:straightforward_algo}). However, if the approximate reliabilities $\tilde \P({\bf x}_s| {\bf y}, s)$ are used instead of the path metrics to select the decoded word $\bf \hat x$ as $\bf \hat x = \arg \max \limits_{s} \tilde \P({\bf x}_s| {\bf y}, s)$, it is possible that the decoded word will not be ML (if the channel is noisy enough).

%
%
\section{Post-Decoding Reliability Computation}
\label{sec:post-proc_rova}

There are $q^\nu$ possible starting states that must be evaluated in the straightforward {\tt TB ROVA} of Sec.~\ref{sec:straightforward_algo} and Fig.~\ref{fig:tb_rova_block}.  Thus it may be beneficial to instead use an existing reduced-complexity tail-biting decoder to find $\bf \hat x$,  and then compute the reliability separately.  Many reduced-complexity tail-biting decoders take advantage of the circular decoding property of tail-biting codes.  Some of these approaches are not maximum likelihood, such as the wrap-around Viterbi algorithm or Bidirectional Viterbi Algorithm (BVA), both discussed in \cite{Shao_TwoTB_TCOM_2003}. The A* algorithm \cite{Shankar_Astar_ISIT_2001,Shankar_Astar_2006,Ortin_Astar_2010} is one ML alternative to the tail-biting decoding method described in Sec. \ref{sec:straightforward_algo}.  Its complexity depends on the SNR.

Suppose that a decoder has already been used to determine $\bf \hat x$ and its starting state $\hat s$, and that we would  like to determine $\P(\bf \hat x|y)$. One operation of {\tt ROVA}$(\hat s)$ would compute the probability $\P({\bf \hat x|y},\hat s)$, but the probability $\P(\hat s|{\bf y})$ required by \eqref{eqn:P_x_y_weighted} would still be undetermined. Must we perform {\tt ROVA}$(s)$ for each $s \neq \hat s$ in order to compute $\P(\hat s | \bf y)$ as in \eqref{eqn:P_s_y}? This section shows how to avoid this by combining the computations of the straightforward approach into a novel Post-decoding Reliability Computation ({\tt PRC}) for tail-biting codes.

Fig.~\ref{fig:post_block} shows a block diagram of {\tt PRC}.  For a rate-$k/n$ tail-biting convolutional code with $\nu$ memory elements, {\tt PRC} takes the following inputs: a received sequence ${\bf y}^L$ with $L$ trellis segments, a \textit{candidate codeword} ${\bf \hat x}^L$ corresponding to a \textit{candidate path} in the trellis, and the starting/ending state $\hat s$ of the candidate codeword. The goal is to compute the posterior probability of the candidate codeword, $\P({\bf \hat x}^L | {\bf y}^L)$. The candidate codeword selected by the decoder may not be the ML codeword.  ~{\tt PRC} computes its true reliability regardless.

\begin{figure}[t]
\centering \def\svgwidth{290pt}
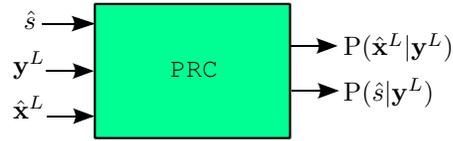
\caption{Block diagram of the Post-decoding Reliability Computation ({\tt PRC}).}
\label{fig:post_block}
\end{figure} 

Raghavan and Baum's ROVA \cite{Raghavan_ROVA_TransIT_1998} performs the traditional add-compare-select operations of the Viterbi algorithm and then computes, for every state in each trellis segment, the posterior probability that the survivor path is correct and the posterior probability that one of the non-surviving paths at the state is correct. Upon reaching the end of the trellis (the $L$th segment), having selected survivor paths at each state, there will be one survivor path corresponding to the ML codeword.

In contrast, with the candidate path already identified, {\tt PRC} processes the trellis without explicitly selecting survivors.  {\tt PRC} computes the reliability of the candidate path and the overall reliability of all other paths.

\subsection{{\tt PRC} Overview}

We define the following events at trellis stage $\ell$ ($\ell \in \{1, 2, \dots, L\}$):
\begin{itemize}
	\item $\p_\ell^{\hat s \rightarrow j} = \{$the candidate path from its beginning at state $\hat s$ to its arrival at state $j$ in segment $\ell$ is correct$\}$
	\item $\bar \p_\ell^{s \rightarrow r} = \{$some path from its beginning at state $s$ to its arrival at state $r$ in segment $\ell$ is correct (including possibly the candidate path if $\hat s=s$) $\}$
	\item $\b_\ell^{i \rightarrow j} = \{$the branch from state $i$ to state $j$ at time $\ell$ is correct$\}$
\end{itemize}

For $\nu = 3$ memory elements, Fig.~\ref{fig:trellis_diagram} gives some examples of the paths corresponding to each of these events.   The black branches in Fig.~\ref{fig:trellis_diagram} constitute all of the paths in the event  $\bar \p_7^{s \rightarrow r}$.   The red branches in Fig.~\ref{fig:trellis_diagram} show the candidate path corresponding to the event $\p_7^{\hat s \rightarrow j}$.  The posterior probability that the red candidate path starting at state $\hat s$ is correct is $\P(\p_7^{\hat s \rightarrow j} | {\bf y}^7)$. The posterior probability that any of the paths that started at state $s$ and arrive at state $r$ in segment $7$ are correct is $\P(\bar \p_7^{r \rightarrow s} | {\bf y}^7)$.
Note that since some branch transitions are invalid in the trellis, $\P(\p_\ell^{\hat s \rightarrow j})$ and $\P(\b_\ell^{i \rightarrow j})$ may be zero for invalid states and branches in segment $\ell$.

The path-correct probabilities $\P(\p_\ell^{\hat s \rightarrow j})$ and $\P(\bar \p_\ell^{r \rightarrow s})$ can be expressed recursively in terms of the probabilities of the previous trellis segments' paths being correct. Conditioned on the noisy channel observations ${\bf y}^\ell = (y_1, y_2, \dots, y_\ell)$, the path-correct probability for the candidate path (which passes through state $i$ in segment $\ell - 1$) is 
\begin{align}
	\P(\p_\ell^{\hat s \rightarrow j} | {\bf y}^\ell) &= \P(\p_{\ell-1}^{\hat s \rightarrow i}, \b_\ell^{i \rightarrow j} | {\bf y}^\ell) & \\
	&= \P(\b_\ell^{i \rightarrow j} | {\bf y}^\ell, \p_{\ell-1}^{\hat s \rightarrow i}) \P(\p_{\ell-1}^{\hat s \rightarrow i} | {\bf y}^\ell) \\
	&= \frac{\mathrm{f}(y_\ell | {\bf y}^{\ell-1}, \p_{\ell-1}^{\hat s \rightarrow i}, \b_\ell^{i \rightarrow j})} {\mathrm{f}(y_\ell | {\bf y}^{\ell-1})} \label{eqn:bayes} \\
 & \times \P(\b_\ell^{i \rightarrow j} | {\bf y}^{\ell-1}, \p_{\ell-1}^{\hat s \rightarrow i})   \P(\p_{\ell-1}^{\hat s \rightarrow i} | {\bf y}^{\ell-1}).   \nonumber
\end{align}
The decomposition in \eqref{eqn:bayes} uses Bayes' rule and follows \cite{Raghavan_ROVA_TransIT_1998}.  Fig.~\ref{fig:trellis_diagram} identifies an example of states $i$ and $j$ used to compute the probability of $\b_7^{i \rightarrow j}$. 

\begin{figure}
  \centering
    	\scalebox{0.56}{\includegraphics{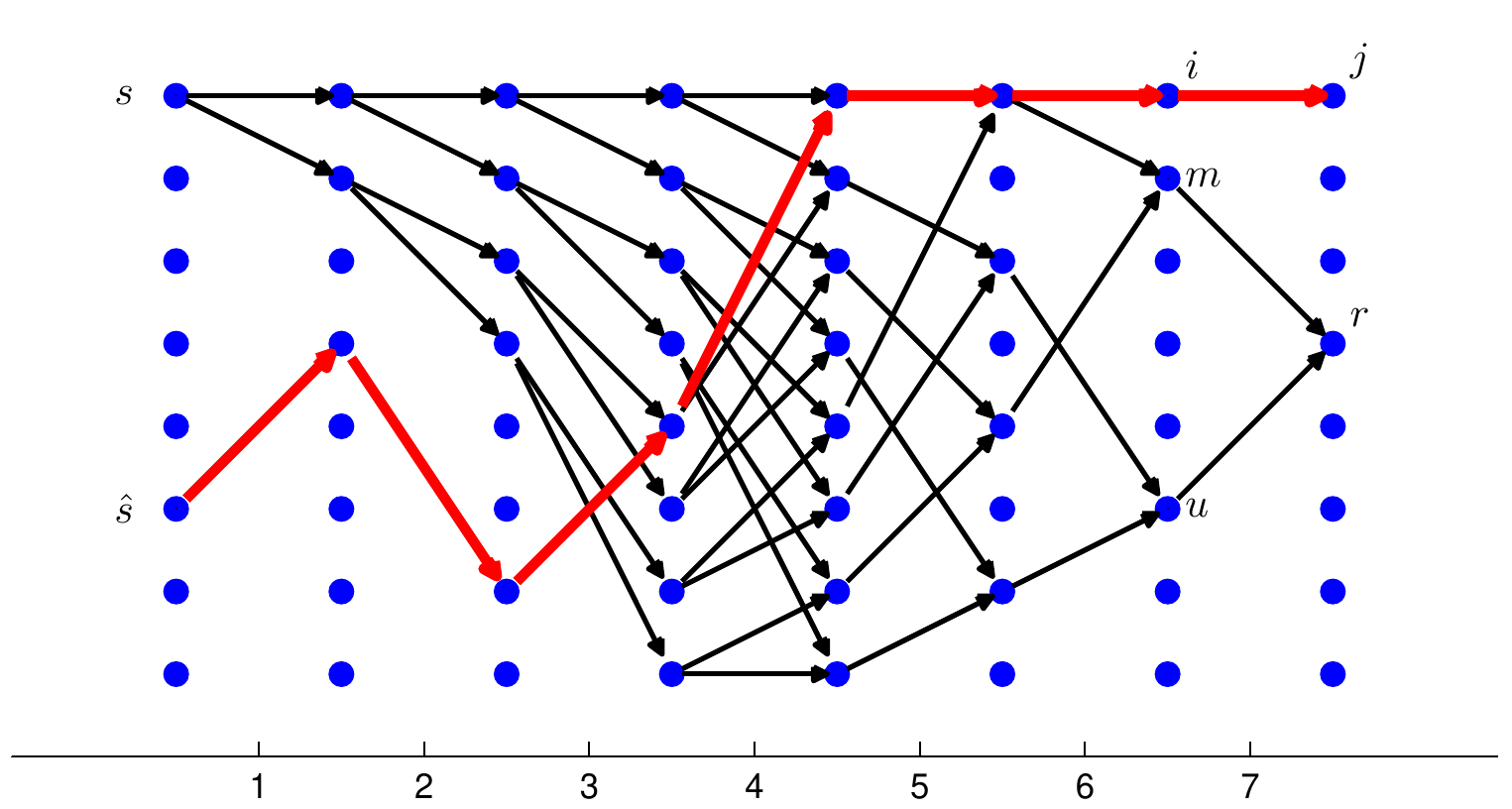}}
\caption{An example of a trellis for a rate-$1/n$ code with $\nu = 3$ memory elements is shown for $\ell = 1, 2, \dots, 7$.  The red branches show the candidate path from its beginning at state $\hat s$ to its arrival at state $j$.  The black branches show all of the paths originating at state $s$ from their beginning to their arrival at state $r$. Note that the figure does not show the final stage of the tail-biting trellis where all paths must return to their starting states.}
\label{fig:trellis_diagram}
\end{figure}

By the Markov property,  $\mathrm{f}(y_\ell | {\bf y}^{\ell-1}, \p_{\ell-1}^{\hat s \rightarrow i}, \b_\ell^{i \rightarrow j}) = \mathrm{f}(y_\ell | \b_\ell^{i \rightarrow j})$, which is the conditional p.d.f., related to the familiar Viterbi algorithm branch metric. Similarly, the second term is $\P(\b_\ell^{i \rightarrow j} | {\bf y}^{\ell-1}, \p_{\ell-1}^{\hat s \rightarrow i}) = \P(\b_\ell^{i \rightarrow j} | \p_{\ell-1}^{\hat s \rightarrow i})$. With these simplifications, \eqref{eqn:bayes} becomes
\begin{align}
	\P(\p_\ell^{\hat s \rightarrow j} | {\bf y}^\ell) &= \frac{\mathrm{f}(y_\ell | \b_\ell^{i \rightarrow j}) \P(\b_\ell^{i \rightarrow j} | \p_{\ell-1}^{\hat s \rightarrow i}) \P(\p_{\ell-1}^{\hat s \rightarrow i} | {\bf y}^{\ell-1})}{\mathrm{f}(y_\ell | {\bf y}^{\ell-1})} .	\label{eqn:P_surv}
\end{align}
%
The denominator can be expressed as a sum over all branches in the trellis $\cal{T}_\ell$ at time $\ell$, where each branch from $m$ to $r$ is denoted by a pair $(m,r) \in \cal{T}_\ell$:
\begin{align}
	\mathrm{f}(y_\ell | {\bf y}^{\ell-1}) &= \sum \limits_{(m,r) \in \cal{T_\ell}} \mathrm{f}(y_\ell | {\bf y}^{\ell-1}, \b_\ell^{m \rightarrow r}) \P(\b_\ell^{m \rightarrow r} | {\bf y}^{\ell-1}) \\
	&= \sum \limits_{(m,r) \in \cal{T_\ell}} \mathrm{f}(y_\ell | \b_\ell^{m \rightarrow r}) \P(\b_\ell^{m \rightarrow r} | {\bf y}^{\ell-1}).
	\label{eqn:f_denom}
\end{align}
The derivation thus far has followed \cite{Raghavan_ROVA_TransIT_1998}, which focused on terminated convolutional codes.

For tail-biting codes, we can further expand the term $\P(\b_\ell^{m \rightarrow r} | {\bf y}^{\ell-1})$ by summing over all the possible starting states $s'$ as follows:
\begin{align}
	\P(\b_\ell^{m \rightarrow r} | {\bf y}^{\ell-1}) &= \sum \limits_{s'} \P(\b_\ell^{m \rightarrow r}, \bar \p_{\ell-1}^{s' \rightarrow m} | {\bf y}^{\ell-1}) \\
	&= \sum \limits_{s'} \P(\b_\ell^{m \rightarrow r} | {\bf y}^{\ell-1},\bar  \p_{\ell-1}^{s' \rightarrow m}) \P(\bar \p_{\ell-1}^{s' \rightarrow m} | {\bf y}^{\ell-1}) \nonumber \\
	&= \sum \limits_{s'} \P(\b_\ell^{m \rightarrow r} | \bar  \p_{\ell-1}^{s' \rightarrow m}) \P(\bar \p_{\ell-1}^{s' \rightarrow m} | {\bf y}^{\ell-1}),
\end{align}
where the last equality follows from the Markov property $\P(\b_\ell^{m \rightarrow r} | {\bf y}^{\ell-1},\bar  \p_{\ell-1}^{s' \rightarrow m})  = \P(\b_\ell^{m \rightarrow r} | \bar  \p_{\ell-1}^{s' \rightarrow m})$.
Thus, \eqref{eqn:f_denom} becomes 
\begin{align}
	\mathrm{f}(y_\ell | {\bf y}^{\ell-1}) = & \sum \limits_{(m,r) \in \cal{T_\ell}} \mathrm{f}(y_\ell | \b_\ell^{m \rightarrow r}) \label{eqn:f_denom_states} \\
	& \times \sum \limits_{s'} \P(\b_\ell^{m \rightarrow r} | \bar  \p_{\ell-1}^{s' \rightarrow m}) \P(\bar \p_{\ell-1}^{s' \rightarrow m} | {\bf y}^{\ell-1}) .	\nonumber
\end{align}

The term $\P(\b_\ell^{m \rightarrow r} | \bar  \p_{\ell-1}^{s' \rightarrow m})$ is the probability that the branch from state $m$ to state $r$ is correct, given that one of the paths that started at state $s'$ and arrived at state $m$ at time $\ell-1$ is correct. 
Recall that $K = L - \nu$.   $\P(\b_\ell^{m \rightarrow r} | \bar  \p_{\ell-1}^{s' \rightarrow m})=q^{-k}$ for $1 \leq \ell \leq K$ (i.e., all $\ell$ except for the last $\nu$ trellis segments).  This is because there are $q^k$ equiprobable next states for these values of $\ell$. 

Using the notation $r \Rightarrow s'$ to indicate there is a valid path from state $r$ at time $\ell$ to state $s'$ at time $L$, we define the following indicator function $\mathrm{I}(\b_\ell^{m \rightarrow r}, s')$, which indicates that the trellis branch from state $m$ to state $r$ at trellis stage $\ell$ is a branch in a possible trellis path that terminates at $s'$:
%
%
\begin{align}
	\mathrm{I}(\b_\ell^{m \rightarrow r}, s')  &= \left \{ \begin{array}{lrll}
	1, & 1 \leq \ell \leq K, & (m,r) \in {\cal T}_\ell, &\\
	1, & K+1 \leq \ell  \leq L, & (m,r) \in {\cal T}_\ell, & r \Rightarrow s'  \\
	0, & K+1 \leq \ell  \leq L, & (m,r) \in {\cal T}_\ell, & r \not\Rightarrow s'  \\
	0, & & (m,r) \not \in {\cal T}_\ell, & \end{array} \right .	\label{eqn:branch_inds}
\end{align}
%
The branch-correct probabilities can now be written as
\begin{align}
	\P(\b_\ell^{m \rightarrow r} | \bar  \p_{\ell-1}^{s' \rightarrow m})  &=  \begin{cases}
	\I(\b_\ell^{m \rightarrow r}, s') ~q^{-k} , & 1 \leq \ell \leq K  \\
	\I(\b_\ell^{m \rightarrow r}, s'), & K+1 \leq \ell  \leq L \end{cases}  \\
	\P(\b_\ell^{m \rightarrow r} | \p_{\ell-1}^{\hat s \rightarrow m})  &=  \begin{cases}
	\I(\b_\ell^{m \rightarrow r}, \hat s) ~q^{-k} , & 1 \leq \ell \leq K  \\
	\I(\b_\ell^{m \rightarrow r}, \hat s), & K+1 \leq \ell  \leq L ~. \end{cases} 
\end{align}
%

We now define the following normalization term for the $\ell$th trellis segment using the above indicators:
\begin{align}
	\Delta_\ell &= \sum \limits_{(m,r) \in \cal{T_\ell}} \f(y_\ell | \b_\ell^{m \rightarrow r})  \sum \limits_{s'} \I(\b_\ell^{m \rightarrow r}, s') \P(\bar \p_{\ell-1}^{s' \rightarrow m} | {\bf y}^{\ell-1}) \label{eqn:delta_ell} \\
	&= \begin{cases}
	\f(y_\ell | {\bf y}^{\ell-1}) ~q^{k} , & 1 \leq \ell \leq  K \\
	\f(y_\ell | {\bf y}^{\ell-1}), & K+1 \leq \ell  \leq L ~. \end{cases}  \label{eqn:delta_ell_f} 
\end{align}
%
The $\Delta_\ell$ normalization term includes most of \eqref{eqn:f_denom_states} but excludes  the potential $q^{-k}$ in $\P(\b_\ell^{m \rightarrow r} | \bar  \p_{\ell-1}^{s' \rightarrow m})$ because it cancels with $\P(\b_\ell^{i \rightarrow j} | \p_{\ell-1}^{\hat s \rightarrow i})$ in the numerator of \eqref{eqn:P_surv}. (Either both have $q^{-k}$ or both are $1$, depending only on $\ell$.)
Substituting \eqref{eqn:delta_ell_f} into \eqref{eqn:P_surv}, we have
\begin{align}
	\P(\p_\ell^{\hat s \rightarrow j} | {\bf y}^\ell) &= \frac{1}{\Delta_\ell} \f(y_\ell | \b_\ell^{i \rightarrow j}) \P(\p_{\ell-1}^{\hat s \rightarrow i} | {\bf y}^{\ell-1}).	\label{eqn:P_surv_delta}
\end{align}
Thus, for the $\ell$th trellis segment, \eqref{eqn:P_surv_delta} expresses the candidate path-correct probability in terms of the candidate path-correct probability in the previous segment.

%

The corresponding expression for the overall path probabilities $\P(\bar \p_\ell^{s \rightarrow r} | {\bf y}^\ell)$ involves more terms. Instead of tracing a single candidate path through the trellis, we must add the probabilities of all the valid tail-biting paths incident on state $r$ in segment $\ell$ as follows:
\begin{align}
	\P(\bar \p_\ell^{s \rightarrow r} | {\bf y}^\ell) = \frac{1}{\Delta_\ell} \sum \limits_{m :(m,r) \in \mathcal{T}_\ell} & \f(y_\ell | \b_\ell^{m \rightarrow r})   \label{eqn:P_nonsurv}	\\
	\times& \I(\b^{m \rightarrow r}_\ell, s ) \P(\bar \p_{\ell - 1}^{s \rightarrow m} | {\bf y}^{\ell-1}). 	\nonumber 
\end{align}
The summation above is over the $q^k$ incoming branches to state $r$. In the special case of a rate-$\frac{1}{n}$ binary code ($q$=$2$), there are $2$ incoming branches, which we will label as $(m,r)$ and $(u,r)$, so \eqref{eqn:P_nonsurv} becomes
\begin{align}
	\P(\bar \p_\ell ^{s \rightarrow r} | {\bf y}^\ell) =& \frac{1}{\Delta_\ell} \big[ \f(y_\ell | \b_\ell^{m \rightarrow r}) \I(\b_\ell^{m \rightarrow r}, s ) \P(\bar \p_{\ell-1} ^{s \rightarrow m} | {\bf y}^{\ell-1}) \nonumber \\
	& + \f(y_\ell | \b_\ell^{u \rightarrow r}) \I(\b_\ell^{u \rightarrow r} ,s ) \P(\bar \p_{\ell-1} ^{s \rightarrow u} | {\bf y}^{\ell-1}) \big].	\label{eqn:P_nonsurv_k1}
\end{align}
Fig.~\ref{fig:trellis_diagram} illustrates how the paths from starting state $s$ merge into state $r$ at trellis segment $\ell=7$.

\subsection{{\tt PRC} Algorithm Summary}

The path probabilities are initialized as follows:
\begin{itemize}
	\item $\P(\p_0^{\hat s \rightarrow j} | {\bf y}^0) = \P(\hat s) = q^{-\nu}$ if $\hat s = j$, or $0$ otherwise. 
	\item $\P(\bar \p_0^{s \rightarrow r} | {\bf y}^0) = \P(s) = q^{-\nu}$ if $s = r$, or $0$ otherwise. 
\end{itemize}

In each trellis-segment $\ell$ ($1 \leq \ell \leq L$), do the following:
\begin{enumerate}
	\item For each branch $(m,r) \in {\cal T}_\ell$, compute the conditional p.d.f. $\mathrm{f}(y_\ell | \b_\ell^{m \rightarrow r})$. 
	\item For each branch $(m,r) \in {\cal T}_\ell$ and each starting state $s$, compute the branch-valid indicator $\I(\b_\ell^{m \rightarrow r} , s)$, as in \eqref{eqn:branch_inds}.
	\item Using the above values, compute the normalization constant $\Delta_\ell$, as in \eqref{eqn:delta_ell}.
	\item For the current state $j$ of the candidate path, compute the candidate path-correct probability $\P(\p_\ell^{\hat s \rightarrow j} | {\bf y}^\ell)$, as in \eqref{eqn:P_surv_delta}.
	\item For each starting state $s$ and each state $r$, compute the overall path-correct probabilities $\P(\bar \p_\ell^{s \rightarrow r} | {\bf y}^\ell)$, as in \eqref{eqn:P_nonsurv}.
\end{enumerate}

After processing all $L$ stages of the trellis, the following meaningful quantities emerge:
\begin{itemize}
	\item The posterior probability that the tail-biting candidate path from $\hat s$ to $\hat s$ is correct is $\P(\p_L^{\hat s \rightarrow \hat s} | {\bf y}^L) = \P(\hat {\bf x}^L | {\bf y}^L)$, which is the probability that the decoded word is correct, given the received sequence.
	\item The posterior word-error probability is then $1~-~\P(\p_L^{\hat s \rightarrow \hat s} | {\bf y}^L)  = 1 - \P(\hat {\bf x}^L | {\bf y}^L)$.
	\item The posterior probability that any of the tail-biting paths (any of the codewords) from $s$ to $s$ is correct is $\P(\bar \p_L^{s \rightarrow s} | {\bf y}^L) = \P(s | {\bf y}^L)$, which is the state reliability desired for \eqref{eqn:P_x_y_weighted}.
\end{itemize}

Numerical results of {\tt PRC} are shown in Fig.~\ref{fig:TB_ROVA_AWGN} in Sec.~\ref{sec:simulation_results}.

%
%
\section{The Tail-biting State-estimation Algorithm}
\label{sec:state_estimation}

The Post-decoding Reliability Computation described above relies on a separate decoder to identify the candidate path.  If, on the other hand, we would like to compute the word-error probability of a tail-biting code without first having determined a candidate path and starting state, we may use the following Tail-Biting State-Estimation Algorithm ({\tt TB SEA}). ~{\tt TB SEA} computes the MAP starting state $\hat s = \arg \max \limits_{s'} \P(s'| {\bf y}^L)$, along with its reliability $\P(\hat s | {\bf y}^L)$. ~{\tt ROVA}$(\hat s)$ can then be used to determine the MAP codeword $\hat {\bf x}_{\hat s}$ corresponding to starting state $\hat s$, as illustrated in Fig.~\ref{fig:tb_sea_block}.

\begin{figure}[tbh]
\centering \def\svgwidth{290pt}
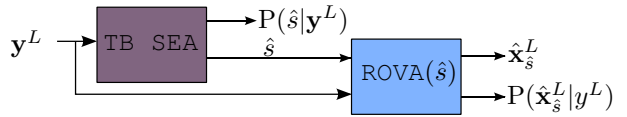
\caption{Block diagram of the Tail-Biting State-Estimation Algorithm ({\tt TB SEA}), followed by {\tt ROVA}$(\hat s)$ for the ML starting state $\hat s$.}
\label{fig:tb_sea_block}
\end{figure} 

{\tt PRC} relied on tracing a single candidate path through the trellis and computing the candidate path-correct probability, as in \eqref{eqn:P_surv_delta}. However, the overall path-correct probabilities in \eqref{eqn:P_nonsurv} do not rely on the candidate path or its probability. The proposed {\tt TB SEA} aggregates all the previous-segment path-correct probabilities $\P(\bar \p_{\ell -1}^{s \rightarrow m} | {\bf y}^{\ell -1})$ as in \eqref{eqn:P_nonsurv}, without regard to a candidate path. As a result, {\tt TB SEA} replaces the traditional add-compare-select operations of the Viterbi algorithm with the addition of all the path probabilities merging into a state that emanate from a particular origin state. Once the entire trellis has been processed, the state reliabilities are compared and the MAP starting state is selected.

\subsection{{\tt TB SEA} Algorithm Summary}
\label{sec:tb_sea}

The path probabilities are initialized as follows:
\begin{itemize}
	\item $\P(\bar \p_0^{s \rightarrow r} | {\bf y}^0) = \P(s) = q^{-\nu}$ if $s = r$, or $0$ otherwise. 
\end{itemize}

In each trellis-segment $\ell$ ($1 \leq \ell \leq L$), do the following:
\begin{enumerate}
	\item For each branch $(m,r) \in {\cal T}_\ell$, compute the conditional p.d.f. $\mathrm{f}(y_\ell | \b_\ell^{m \rightarrow r})$. 
	\item For each branch $(m,r) \in {\cal T}_\ell$ and each starting state $s$, compute the branch-valid indicator $\I(\b_\ell^{m \rightarrow r} , s)$, as in \eqref{eqn:branch_inds}.
	\item Using the above values, compute the normalization constant $\Delta_\ell$, as in \eqref{eqn:delta_ell}.
	\item For each starting state $s$ and each state $r$, compute the overall path-correct probabilities $\P(\bar \p_\ell^{s \rightarrow r} | {\bf y}^\ell)$, as in \eqref{eqn:P_nonsurv}.
\end{enumerate}

After processing all $L$ stages of the trellis, the following meaningful quantity emerges:
\begin{itemize}
	\item The posterior probability that any of the tail-biting paths (any of the codewords) from $s$ to $s$ is correct is $\P(\bar \p_L^{s \rightarrow s}  | {\bf y}^L) = \P(s | {\bf y}^L)$.
\end{itemize}
{\tt TB SEA} selects the starting state with the maximum value of $\P(s | {\bf y}^L)$ (the MAP choice of starting state), yielding $\hat s$ and its reliability $\P(\hat s| {\bf y}^L)$. Thus, {\tt TB SEA} has selected the MAP starting state without explicitly evaluating all possible codewords (i.e., paths through the trellis). This result is not limited to error control coding; it can be applied in any context to efficiently compute the MAP starting state of a tail-biting, finite-state Markov process. 

\subsection{\tt{TB SEA + ROVA}$(\hat s)$}

After finding the MAP starting state $\hat s$ with {\tt TB SEA}, {\tt ROVA}$(\hat s)$ may be used to compute the MAP codeword $\hat {\bf x}_{\hat s}^L$ and $\P(\hat {\bf x}_{\hat s}^L | {\bf y}^L, \hat s)$. We have used the subscript $\hat s$ to indicate that $\hat {\bf x}_{\hat s}^L$ is the MAP codeword for the terminated code starting and ending in $\hat s$. The overall reliability $\P(\hat {\bf x}_{\hat s}^L | {\bf y}^L)$ can then be computed as in \eqref{eqn:P_x_y_weighted}, which we have replicated below to show how the {\tt TB SEA} and {\tt ROVA}$(\hat s)$ provide the needed factors:

\begin{align}
	\P(\bf \hat x|y) &= \underbrace{\P(\hat {\bf x}|{\bf y},\hat s)}_\text{computed by {\tt ROVA$(\hat s)$}} \times \underbrace{\P(\hat s|y).}_\text{computed by {\tt TB SEA}}
	\label{eqn:P_x_y_TBSEA}
\end{align}

%
%
\subsection{MAP States vs. MAP Codewords}

Is it possible that the maximum a posteriori codeword $\bf \hat x$ corresponds to a starting state other than the MAP state $\hat s$? The following theorem proves that the answer is no, given a suitable probability of error.

\begin{theorem}
	The MAP codeword $\bf \hat x$ for a tail-biting convolutional code begins and ends in the MAP state $\hat s$, as long as $\P({\bf \hat x | y}) > \frac{1}{2}$. \label{thm:MLstates}
\end{theorem}
\begin{proof}
	Consider a codeword $\bf \hat x$ with  $\P({\bf \hat x | y}) >  \frac{1}{2}$. By \eqref{eqn:P_x_y_weighted}, $\P({\bf \hat x | y}) = \P({\bf \hat x | y}, s_{\bf \hat x}) \P(s_{\bf \hat x} | \bf y)$, where $s_{\bf \hat x}$ is the starting state of $\bf \hat x$. This implies that $\P(s_{\bf \hat x} | {\bf y}) > \frac{1}{2}$. The MAP state is $\arg \max \limits_{s'} \P(s' | \bf y)$, which must be $s_{\bf \hat x}$, since all other states $s'$ must have  $\P(s' | {\bf y}) < \frac{1}{2}$.
\end{proof}

Theorem~\ref{thm:MLstates} shows that the application of {\tt TB SEA} followed by the Viterbi algorithm (or the ROVA) will always yield the MAP codeword $\bf \hat x$ of the tail-biting code, not just the MAP codeword for the terminated code starting in $\hat s$ (as long as the probability of error is less than $\frac{1}{2}$). In most practical scenarios, the word-error probability ($1 -  \P(\bf \hat x | y)$),  even if unknown exactly, is much less than $\frac{1}{2}$, so the theorem holds. As a result, in these cases {\tt TB SEA} selects the same codeword $\bf \hat x$ as would the {\tt TB ROVA} of Sec.~\ref{sec:straightforward_algo}, and computes the same reliability $\P(\bf \hat x|y)$.

\begin{center}
\begin{table*}[bp]
{\small
\hfill{}
\begin{center}
  \caption{Complexity per trellis segment of the proposed algorithms (disregarding branch metric computations).} 
\begin{tabular}{ l | c | c | c | c | c | c }
 Algorithm & \rot{Path metrics} & \rot{Cand. prob.} & \rotatebox{45}{Overall prob.} & \rotatebox{45}{Additions} & \rotatebox{45}{Multiplications} & \rotatebox{45}{Divisions} \\
  \hline 
    \multicolumn{7}{c}{Key modules of decoders}  \\
  \hline
 {\tt VA}$(s)$ & $q^\nu$ & $0$ & $0$ & $0$ & $q^{\nu + k}$ & $0$ \\
 {\tt ROVA}$(s)$ \cite{Raghavan_ROVA_TransIT_1998} & $q^\nu$ & $q^\nu$ & $q^\nu$ & $2q^{\nu}(2 q^k - 1) - 1$ & $3q^{\nu + k}$ & $2q^{\nu}$ \\
 {\tt PRC} & $0$ & $1$ & $q^{2 \nu}$ &  $q^{2 \nu}(2 q^k - 1) - 1$ & $q^{2 \nu + k}$ & $q^{2 \nu} + 1$ \\
 \tt{TB SEA} & $0$ & $0$ & $q^{2 \nu}$ &  $q^{2 \nu}(2 q^k - 1) - 1$ & $q^{2 \nu + k}$ & $q^{2 \nu}$ \\
 {\tt Approx ROVA}$(s)$ \cite{Fricke_Approx_ROVA_VTC_2007} & $q^\nu$ & $q^\nu$ & $0$ & $q^{\nu}(q^k-1)$ & $q^{\nu + k} + 1$ & $q^{\nu}$ \\
  \hline
    \multicolumn{7}{c}{Tail-biting decoders that provide reliability output}  \\
  \hline
 \tt{TB ROVA} & $q^{2 \nu}$ & $q^{2 \nu}$  & $q^{2 \nu}$  &  $2q^{2 \nu}(2 q^k - 1) - q^\nu$ & $3q^{2 \nu + k}$ & $2q^{2 \nu}$ \\
 \tt{TB SEA + ROVA$(\hat s)$} & $q^\nu$ & $q^\nu$ & $q^{2 \nu} + q^\nu$ & $(q^{2 \nu} + 2q^\nu)(2q^k - 1) - 2$ & $q^{2 \nu + k} + 3 q^{\nu+k}$ & $q^{2 \nu} + 2q^\nu$ \\
 \tt{Approx TB ROVA} & $q^{2 \nu}$ & $q^{2 \nu}$ & $0$ & $q^{2 \nu}(q^k-1)$ & $q^{2 \nu + k} + q^\nu$ & $q^{2 \nu}$ \\
\end{tabular}
\label{tbl:complexity}
\end{center}}
\hfill{}
\end{table*}
\end{center}

%
%
\subsection{The TB BCJR Algorithm for State Estimation}
\label{sec:tbbcjr}

While several related papers such as \cite{Handlery_Boosting_TCOM_2003} and \cite{Yu_State_Est_VTC_2008} have proposed ways to estimate the starting state of a tail-biting decoder, none computes exactly the posterior probability of the starting state, $\P(s| {\bf y}^L)$, as described for {\tt TB SEA}. Upon a first inspection, the tail-biting BCJR (TB BCJR) of \cite{Anderson_TB_MAP_JSAC_1998} appears to provide a similar method of computing this probability. Applying the forward recursion of the BCJR algorithm provides posterior probabilities that are denoted as $\alpha_L(s) = \P(s| {\bf y}^L)$ in \cite{Anderson_TB_MAP_JSAC_1998}. Thus, it would appear that the state-estimation algorithm of Sec.~\ref{sec:tb_sea} can be replaced by a portion of the TB BCJR algorithm. This would yield a significant decrease in computational complexity, from roughly $q^{2 \nu}$ operations per trellis segment for {\tt TB SEA} to $q^\nu$ for the TB BCJR. However, as will be shown in Sec.~\ref{sec:simulation_results}, the word-error performance of the tail-biting BCJR when used in this manner is significantly inferior to that of {\tt TB SEA}.

As noted in \cite{Anderson_TB_BCJR_book_2001} and \cite{Anderson_BD_CVA_TCOM_2002}, the tail-biting BCJR algorithm is an approximate symbol-by-symbol MAP decoder. It is approximate in the sense that the forward recursion of the TB BCJR in \cite{Anderson_TB_MAP_JSAC_1998} does not strictly enforce the tail-biting restriction, allowing non-tail-biting ``pseudocodewords'' to appear and cause errors. \cite{Anderson_TB_MAP_JSAC_1998} requires the probability distributions of the starting and ending states to be the same, which is a weaker condition than requiring all codewords to start and end in the same state. \cite{Anderson_TB_BCJR_book_2001} and \cite{Anderson_BD_CVA_TCOM_2002} have shown that when the tail-biting length $L$ is large relative to the memory length $\nu$, the suboptimality of the TB BCJR in terms of the bit-error rate is small. However, we are concerned with word-error performance in this paper. We find that when the TB BCJR is used to estimate the initial state $\hat s$ and its probability $\P(\hat s| {\bf y}^L)$, followed by {\tt ROVA}$(\hat s)$ for the most likely state $\hat s$, the impact on word error is severe (Fig.~\ref{fig:TB_ROVA_AWGN}). Frequent state-estimation errors prevent the Viterbi algorithm in the second phase from decoding to the correct codeword. Thus, the approximate tail-biting BCJR of \cite{Anderson_TB_MAP_JSAC_1998} is not effective as a replacement for {\tt TB SEA} when using the word-error criterion.

In contrast, the exact symbol-by-symbol MAP decoder for tail-biting codes in \cite[Ch. 7]{Johannesson_Fundamentals_Conv_1999} does enforce the tail-biting restriction, and has complexity on the same order as that of {\tt TB SEA}. However, because the symbol-by-symbol MAP decoder selects the most probable input symbols while {\tt TB SEA + ROVA}$(\hat s)$ selects the most probable input sequence, {\tt TB SEA + ROVA}$(\hat s)$ is recommended for use in retransmission schemes that depend on the word-error probability.

%
%
\section{Complexity Analysis}
\label{sec:complexity}

Table~\ref{tbl:complexity} compares the complexity per trellis segment of each of the discussed algorithms, assuming that the conditional p.d.f. $\mathrm{f}(y_\ell | \b_\ell^{m \rightarrow r})$ has already been computed for every branch in the trellis.
The columns labeled  `Path metrics', `Cand. prob.', and `Overall prob.' refer to the number of quantities that must be computed and stored in every trellis segment, for the path metrics of the Viterbi algorithm, the candidate path probability of \eqref{eqn:P_surv_delta}, and the overall path probability of \eqref{eqn:P_nonsurv}, respectively. The number of operations per trellis segment required to compute these values is listed in the columns labeled `Additions', `Multiplications', and `Divisions'.

The {\tt ROVA($s$)} row of Table~\ref{tbl:complexity} corresponds to Raghavan and Baum's ROVA \cite{Raghavan_ROVA_TransIT_1998} for a terminated code starting and ending in state $s$. The operations listed include the multiplications required for the path metric computations of the Viterbi algorithm for state $s$, {\tt VA}$(s)$.
The {\tt TB ROVA} row represents performing the ROVA for each of the $q^\nu$ possible starting states as described in Sec.~\ref{sec:straightforward_algo}, so each of the quantities is multiplied by $q^\nu$.

The {\tt PRC} row corresponds to the proposed Post-decoding Reliability Computation of Sec.~\ref{sec:post-proc_rova}. The complexity incurred to determine the candidate path (e.g., by the BVA or the A* algorithm) is not included in this row and must also be accounted for, which is why no path metrics are listed for {\tt PRC}. 
Compared to {\tt TB ROVA}, due to combining computations into a single pass through the trellis, the complexity of {\tt PRC} is reduced by approximately a factor of 2. This is because {\tt TB ROVA} calculates a candidate path probability for each of the $q^\nu$ starting states (due to decoding to the ML codeword each time), whereas the combined trellis-processing of {\tt PRC} involves only one candidate path. Both algorithms compute $q^\nu$ overall path probabilities, so the ratio of complexity is roughly $\frac{1 + q^\nu}{q^\nu + q^\nu} \approx \frac{1}{2}$. 

The reduction in complexity of {\tt TB SEA} compared to {\tt PRC} is modest, with slightly fewer multiplications and divisions required due to the absence of the candidate path calculations in {\tt TB SEA}.
Importantly, performing {\tt TB SEA} followed by {\tt ROVA}$(\hat s)$ for the ML state $\hat s$ is shown to be an improvement over {\tt TB ROVA} for moderately large $\nu$. ~{\tt TB SEA + ROVA}$(\hat s)$ requires approximately one half the additions, one third the multiplications, and one half the divisions of {\tt TB ROVA}.
{\tt TB SEA}'s complexity reduction is partly due to the fact that it does not require the add-compare-select operations of the Viterbi algorithm, which {\tt TB ROVA} performs for each starting state. 
Note also that the number of trellis segments processed in {\tt TB SEA} is constant ($L$ segments), whereas the number of trellis segments processed by many tail-biting decoders (e.g., the BVA) depends on the SNR.

Lastly, the computational costs of performing Fricke and Hoeher's simplified ROVA \cite{Fricke_Approx_ROVA_VTC_2007} are listed in the {\tt Approx ROVA}$(s)$ row, along with the tail-biting approximate version of Sec.~\ref{sec:simplified} ({\tt Approx TB ROVA}). In both of these cases, the word-error outputs are estimates. In contrast, {\tt TB ROVA} and {\tt TB SEA + ROVA}$(\hat s)$ compute the exact word-error probability of the received word. 

For the special case of rate-$1/n$ binary convolutional codes with $\nu = 6$ memory elements, Fig.~\ref{fig:complexity} gives an example of the number of additions, multiplications, and divisions that must be performed per trellis segment for the three tail-biting decoders in Table~\ref{tbl:complexity}. ~{\tt TB SEA + ROVA}$(\hat s)$ is competitive with {\tt Approx TB ROVA} in terms of the number of multiplications and divisions that must be performed, but {\tt Approx TB ROVA} requires fewer additions than {\tt TB SEA + ROVA}$(\hat s)$.
%

\begin{figure}
  \centering
    	\scalebox{0.47}{\includegraphics{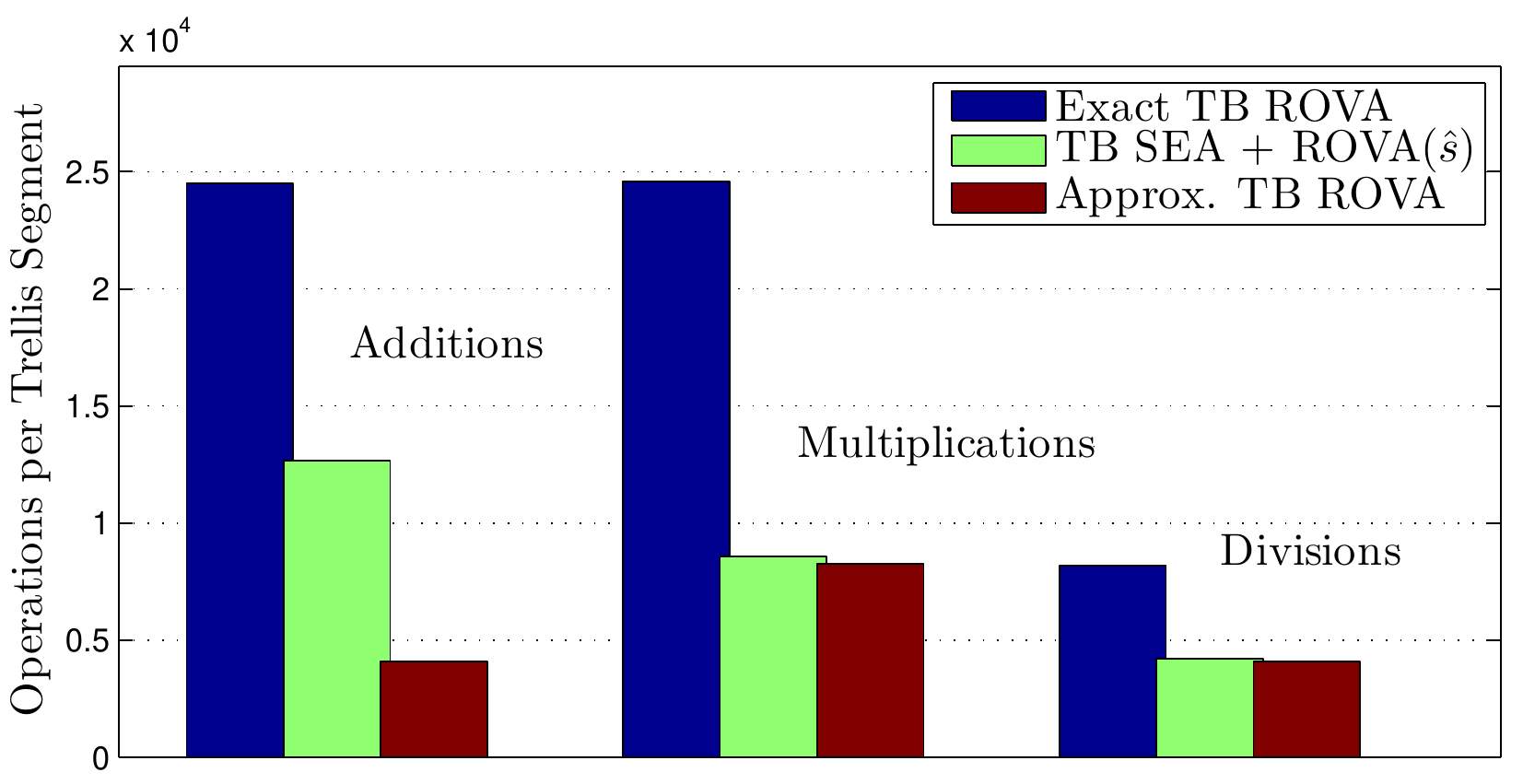}}
    \hspace{0.1in}
\caption{Examples of the computations per trellis segment for the tail-biting decoders listed in Table~\ref{tbl:complexity} corresponding to rate-$1/n$ binary convolutional codes, for $\nu = 6$ memory elements.}
\label{fig:complexity}
\end{figure}

%
%
\section{Numerical Results}
\label{sec:simulation_results}

\subsection{Convolutional Code}

\begin{table}
\begin{center}
  \caption{Generator polynomials $g_1$, $g_2$, and $g_3$ corresponding to the simulated rate-$1$/$3$ tail-biting convolutional code. $d_{free}$ is the free distance, $A_{d_{free}}$ is the number of nearest neighbors with weight $d_{free}$, and $L_D$ is the analytic traceback depth.} 
\begin{tabular}{ c | c | c | c | c | c | c | c }
  $\nu$ & $2^{\nu}$ & $g_1$ & $g_2$ & $g_3$  & $d_{free}$ & $A_{d_{free}}$ & $L_D$ \\
  \hline \hline
  6 & 64 & 117 & 127 & 155 & 15 & 3 & 21 \\
  \hline
\end{tabular}
\label{tbl:conv_codes}
\end{center}
\end{table}

Table \ref{tbl:conv_codes} lists the rate-$1$/$3$, binary convolutional encoder polynomials from Lin and Costello \cite[Table 12.1]{Lin_2004_ECC} used in the simulations. The number of memory elements is $\nu$, $2^\nu$ is the number of states, and $\{g_1,g_2,g_3\}$ are the generator polynomials in octal notation. The code selected has the optimum free distance $d_{free}$, which is listed along with the analytic traceback depth $L_D$ \cite{Anderson_Traceback_TransIT_1989}. $A_{d_{free}}$ is the number of nearest neighbors with weight $d_{free}$. The simulations in this section use a feedforward encoder realization of the generator polynomial.

%
\subsection{Additive White Gaussian Noise (AWGN) Channel}

For antipodal signaling (i.e., BPSK) over the Gaussian channel, the conditional density $\mathrm{f}(y_\ell | \b_\ell^{i \rightarrow j})$ can be expressed as

\begin{align}
	\mathrm{f}(y_\ell | \b_\ell^{i \rightarrow j}) = \prod \limits_{m=1}^n \frac{1}{\sqrt{2 \pi \sigma^2}} \exp \left \{ - \frac{[y_\ell(m) - x_\ell(m) ]^2 }{2 \sigma^2} \right \},
\end{align}
where $y_\ell(m)$ is the $m$th received BPSK symbol in trellis-segment $\ell$, $x_\ell(m)$ is the $m$th output symbol of the encoder branch from state $i$ to state $j$ in trellis segment $\ell$, and $\sigma^2$ is the noise variance. For a transmitter power constraint $P$, the encoder output symbols are $x_\ell(m) \in \{ +\sqrt{P},-\sqrt{P} \}$ and the energy per bit is $E_b = P \frac{n}{k}$. This yields an SNR equal to $P/\sigma^2 = 2 \frac{k}{n} \frac{E_b}{N_0}$ when the noise variance is $\sigma^2~=~N_0/2$.

\subsection{Simulation Results}

This section provides a comparison of the average word-error probability computed by the tail-biting reliability-output algorithms for the AWGN channel using the rate-$1/3$, $64$-state tail-biting convolutional code listed in Table~\ref{tbl:conv_codes}. 
The simulations in Fig.~\ref{fig:TB_ROVA_AWGN} use $L = 128$ input bits and $384$ output bits. The `Actual' curves in the figures show the fraction of codewords that are decoded incorrectly, whereas the `Computed' curves show the word-error probability computed by the receiver.
`Actual' values are only plotted for simulations with more than 100 codewords in error. 

Fig.~\ref{fig:TB_ROVA_AWGN} evaluates the performance of Sec.~\ref{sec:straightforward_algo}'s {\tt TB ROVA} and Sec.~\ref{sec:post-proc_rova}'s {\tt PRC}. In the figure, {\tt PRC} is applied to the output of the Bidirectional Viterbi Algorithm (BVA), a suboptimal tail-biting decoder. 
The `Actual' word-error performance of the suboptimal `BVA'  is slightly worse than that of the ML `Exact TB ROVA', but the difference is not visible in Fig. 7. 
However, even though the bidirectional Viterbi decoder may choose a codeword other than the ML codeword, the posterior probability $P(\hat {\bf x}^L|{\bf y}^L)$ computed by {\tt PRC} is exact. 
Thus, {\tt PRC} provides reliability information about the decoded word that the receiver can use as retransmission criteria in a hybrid ARQ setting.

Fig.~\ref{fig:TB_ROVA_AWGN} also shows the performance of the combined {\tt TB SEA + ROVA}$(\hat s)$ approach in comparison with {\tt TB ROVA}. As shown in Thm.~\ref{thm:MLstates}, the word-error probability calculated by the computationally efficient {\tt TB SEA + ROVA}$(s)$ is identical to that of {\tt TB ROVA}, except when the probability of error is extremely high (i.e., when $\P({\bf \hat x|y}) < \frac{1}{2}$). Even in the high-error regime, however, the difference is negligible.

\begin{figure}
  \centering
    	\scalebox{0.47}{\includegraphics{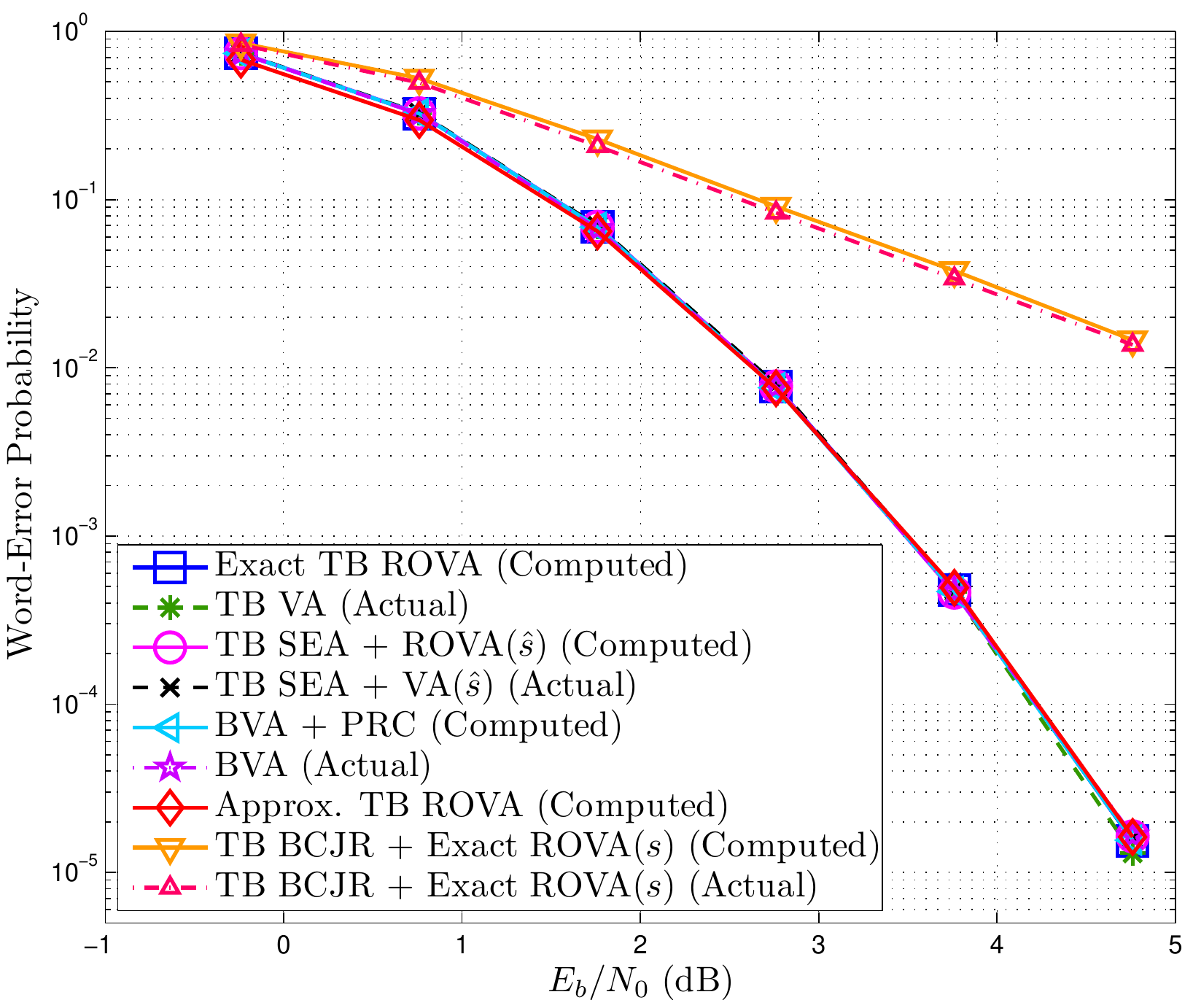}}
    \hspace{0.1in}
\caption{Computed and actual word-error probability of the exact {\tt TB ROVA}, {\tt TB SEA} followed by  {\tt ROVA}$(\hat s)$, and the Bidirectional Viterbi Algorithm (BVA) followed by the Post-decoding Reliability Computation ({\tt PRC}). Also shown are the computed word-error probability estimates for {\tt Approx TB ROVA}.
All simulations use the rate-$1/3$, $64$-state convolutional code listed in Table~\ref{tbl:conv_codes} with $L = 128$ input bits and $384$ output bits and transmission over the AWGN channel. The `Computed' values are the word-error probabilities calculated by the receiver (averaged over the simulation) and the `Actual' values count the number of words decoded incorrectly. 
The `TB BCJR' method of estimating the initial state is included for comparison, indicating that there is a severe penalty for disregarding the tail-biting restriction.  Note that all curves except for the two `TB BCJR' curves are almost indistinguishable.}
\label{fig:TB_ROVA_AWGN}
\end{figure}

The performance of the exact and approximate versions of {\tt TB ROVA} is compared in Fig.~\ref{fig:TB_ROVA_AWGN}. For each starting state $s$, the `Exact TB ROVA' uses Raghavan and Baum's ROVA \cite{Raghavan_ROVA_TransIT_1998} and the `Approx. TB ROVA' uses Fricke and Hoeher's simplified ROVA \cite{Fricke_Approx_ROVA_VTC_2007}, as described in Sec.~\ref{sec:simplified}. The approximate approach results in an estimated word-error probability that is very close to the exact word-error probability. Both reliability computations invoke the same decoder, the tail-biting Viterbi algorithm (`TB VA'), so the `Actual' curves are identical.

Finally, Fig.~\ref{fig:TB_ROVA_AWGN} also shows that when the forward recursion of the `TB BCJR' of \cite{Anderson_TB_MAP_JSAC_1998} is used to estimate the starting/ending state, there is a severe word-error penalty for disregarding the tail-biting restriction, as discussed in Sec.~\ref{sec:tbbcjr}. The `TB BCJR' simulations used one iteration through $L=128$ trellis segments. Simulations with additional loops around the circular trellis did not improve the actual word-error probability, since the tail-biting condition was not enforced. Care should be taken when estimating the starting-state probability $\P(s|{\bf y}^L)$ based on observations of ${\bf y}^L$ in multiple trellis-loops.

\begin{figure}
  \centering
    	\scalebox{0.47}{\includegraphics{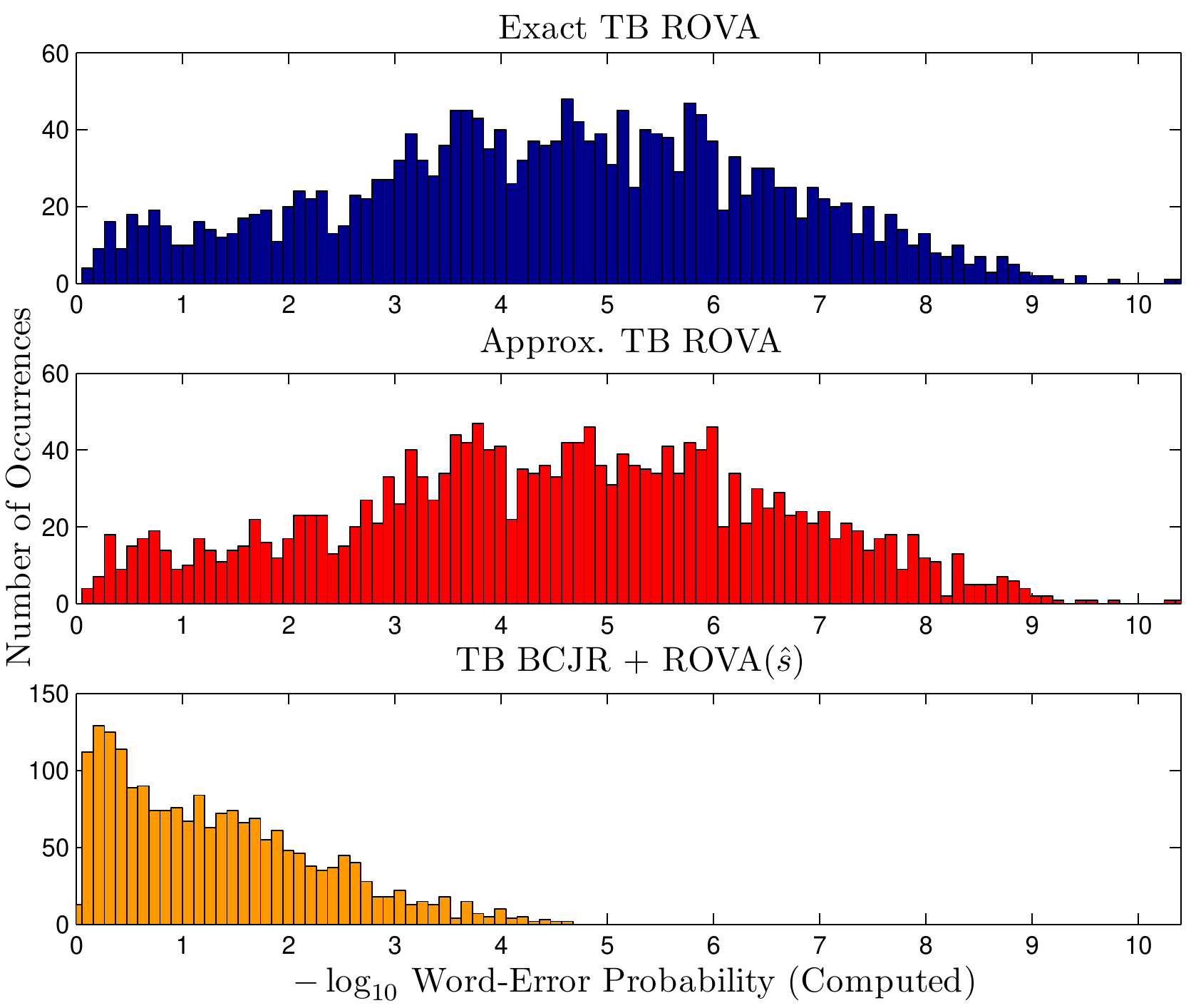}}
    \hspace{0.1in}
\caption{Histograms of the word-error probability, plotted on a logarithmic scale, computed by three reliability-output decoders: {\tt TB ROVA} with ML decoding and exact reliability computations, {\tt Approx TB ROVA} with ML decoding and approximate reliability computations, and the TB BCJR for sub-optimal estimation  of the starting state $\hat s$ followed by {\tt ROVA}$(\hat s)$. Each histogram includes simulations of the same 2000 transmitted codewords and noise realizations. The vertical axis is the number of times among the 2000 decoded words that the word-error probability falls within the histogram bin.                                                                                                                                                                                                                                                                                                                                                                                                                                                                                                                                                                                                                                   
The two ML decoders compute $\P(\hat {\bf x}^L|{\bf y}^L)$ for the same decoded word $\hat {\bf x}^L$, whereas the suboptimal BCJR-based decoder decodes to a codeword that is not necessarily the same as $\hat {\bf x}^L$.
All simulations use the rate-$1/3$, $64$-state convolutional code listed in Table~\ref{tbl:conv_codes}, with $L=32$ input bits, $96$ output bits and SNR 0 dB ($E_b/N_0 = 1.76$ dB).}
\label{fig:TB_ROVA_histogram}
\end{figure}

Fig.~\ref{fig:TB_ROVA_histogram} provides a histogram of the word-error probabilities computed by the receiver for the rate-$1/3$, $64$-state convolutional code listed in Table~\ref{tbl:conv_codes}, with $L=32$ input bits, $96$ output bits and SNR 0 dB ($E_b/N_0 = 1.76$ dB). Fig.~\ref{fig:TB_ROVA_histogram} illustrates that the exact and approximate {\tt TB ROVA} approaches give very similar word-error probabilities, whereas the word-error probabilities computed by the tail-biting BCJR followed by {\tt ROVA}$(\hat s)$ differ significantly. 
The difference in the histogram for the TB BCJR is due to poorer decoder performance.  Frequent errors in the state-estimation portion of the tail-biting BCJR cause the word-error probability to be high.

%
%
\section{Conclusion}
\label{sec:conc}

We have extended the reliability-output Viterbi algorithm to accommodate tail-biting codes, providing several tail-biting reliability-output decoders. {\tt TB ROVA} invokes Raghavan and Baum's ROVA for each possible starting state $s$, and then computes the posterior probability of the ML starting state, $\P(\hat s|{\bf y})$, in order to compute the overall word-error probability. We then demonstrated an approximate version of {\tt TB ROVA} using Fricke and Hoeher's simplified ROVA.
We introduced the Post-decoding Reliability Computation, which calculates the word-error probability of a decoded word, and the Tail-Biting State-Estimation Algorithm, which first computes the MAP starting state $\hat s$ and then decodes based on that starting state with {\tt ROVA}$(\hat s)$.  

A complexity analysis shows that {\tt TB SEA} followed by {\tt ROVA}$(\hat s)$ reduces the number of operations by approximately half compared to {\tt TB ROVA}. Importantly, Theorem 1 proved that the word-error probability computed by {\tt TB SEA + ROVA}$(\hat s)$ is the same as that computed by {\tt TB ROVA} in SNR ranges of practical interest. Because of this, {\tt TB SEA} is a suitable tail-biting decoder to use in reliability-based retransmission schemes (i.e., hybrid ARQ), being an alternative to {\tt Approx TB ROVA}.

%
%
\bibliographystyle{IEEEtran}
{\bibliography{AW_bib}}

\newpage

\end{document}